\documentclass[aps,physrev,twocolumn,superscriptaddress]{revtex4-2}
\usepackage[utf8]{inputenc}
\usepackage{times}
\usepackage{mathtools}
\usepackage{amsmath,amsfonts,amssymb}
\usepackage{array,graphicx,booktabs}
\usepackage[utf8]{inputenc}
\usepackage[T1]{fontenc}
\usepackage{url}
\usepackage{theorem}

\theoremstyle{break}
\theorembodyfont{\itshape}
\newtheorem{theorem}{Theorem}

\theorembodyfont{\normalfont}
\newtheorem{definition}[theorem]{Definition}

\theorembodyfont{\itshape}

\newenvironment{proof}[1][Proof]{\par\noindent\textit{#1.}\ }{\hfill $\square$\par}

\newcommand{\Ftwo}{\mathbb F_2}

\newcommand{\GL}{\operatorname{GL}}

\allowdisplaybreaks

\newcounter{appsection}
\renewcommand{\theappsection}{\Alph{appsection}}
\newcounter{appsubsection}[appsection]
\renewcommand{\theappsubsection}{\theappsection.\arabic{appsubsection}}
\newcommand{\appsection}[2]{%
    \refstepcounter{appsection}%
    \section*{Appendix~\theappsection: #2}%
    \label{#1}%
}
\newcommand{\appsubsection}[2]{%
    \refstepcounter{appsubsection}%
    \subsection*{\texorpdfstring{\theappsubsection\quad #2}{\theappsubsection\space #2}}%
    \label{#1}%
}

\usepackage[
    breaklinks,
    pdftex,
    colorlinks=true,
    linkcolor=blue,
    citecolor=blue,
    urlcolor=blue
]{hyperref}
\usepackage{cleveref}

\begin{document}

\title{Logical Entangling with Phantom Codes in Hypergraph Products}

\author{Keming He}
\thanks{These authors contributed equally to this work.}
\affiliation{QudeLeap Research, Shanghai 200030, China.}
\affiliation{Thrust of Artificial Intelligence, Information Hub, \\ The Hong Kong University of Science and Technology (Guangzhou), Guangdong 511453, China.}

\author{Ziao Tang}
\thanks{These authors contributed equally to this work.}
\affiliation{Thrust of Artificial Intelligence, Information Hub, \\ The Hong Kong University of Science and Technology (Guangzhou), Guangdong 511453, China.}
\author{Zetong Li}
\affiliation{Thrust of Artificial Intelligence, Information Hub, \\ The Hong Kong University of Science and Technology (Guangzhou), Guangdong 511453, China.}
\author{Ge Bai}
\email{gebai@hkust-gz.edu.cn}
\author{Xin Wang}
\email{felixxinwang@hkust-gz.edu.cn}
\affiliation{Thrust of Artificial Intelligence, Information Hub, \\ The Hong Kong University of Science and Technology (Guangzhou), Guangdong 511453, China.}

\begin{abstract}
Logical entangling gates are a major source of physical spacetime overhead in fault-tolerant quantum computation. Phantom codes reduce this cost by implementing every ordered in-block logical CNOT through physical qubit permutations and Pauli-frame updates. Whether this mechanism can coexist with the low-weight stabilizer structure of qLDPC codes is a central question for low-overhead fault-tolerant architectures. We give a deterministic answer within binary CSS hypergraph product (HGP) codes. Up to natural equivalences, the simplex–repetition family is the unique HGP family satisfying the phantom condition. We then evaluate this family under circuit-level noise in logical GHZ-state preparation and Trotterized many-body quantum simulation. The codes retain low-weight stabilizer checks and yield concrete advantages over rotated surface-code baselines in both benchmarks. Reconfigurable neutral-atom arrays offer a natural setting for this approach, supporting nonlocal qLDPC operations while enabling in-block logical CNOTs without additional physical operations. Together, these results make precise how permutation-based logical entangling constrains code design within the HGP framework, demonstrate the circuit-level benefits of the unique family, and guide the search for phantom qLDPC families with better asymptotic parameters for low-overhead fault tolerance on neutral-atom hardware.
\end{abstract}
\maketitle


\emph{Introduction.}---
Quantum error correction has recently achieved major experimental advances
across multiple hardware platforms, particularly superconducting,
neutral-atom, and trapped-ion systems
\cite{tham2026breakeven,perlin2026fault,lin2026surface,
computing2026quantum,wang2026demonstration,rosenfeld2025magic,
dasu2026computing,google2025belowthreshold,bluvstein2024logical,
bluvstein2026fault,lacroix2025scaling,putterman2025hardware,
ryan2024high}.
These advances are bringing fault-tolerant logical devices to the center of
quantum computing research. Fault tolerance enables long computations by
encoding information into quantum error-correcting codes and repeatedly
correcting errors before they accumulate into logical faults
\cite{Shor1996FaultTolerant,AharonovBenOr2008Constant}.
In this regime, the physical spacetime volume required to implement encoded
computation becomes a central component of its cost. Syndrome extraction,
decoding, routing, idle protection, and additional error-correction rounds all
contribute to this overhead. Reducing these physical resources is therefore
essential for scaling fault-tolerant quantum computation.

A major approach to this challenge develops sparse quantum codes with
low-weight stabilizer checks, favorable encoding rates, and scalable decoding,
most notably within the qLDPC program
\cite{bravyi2014homological,panteleev2021quantum,
panteleev2022asymptotically,breuckmann2021balanced,
leverrier2022quantum,dinur2023good}.
Within this landscape, hypergraph product~(HGP) codes provide a particularly natural
setting. They construct sparse CSS stabilizer codes from classical
parity-check matrices and admit a transparent algebraic description of their
logical operator spaces
\cite{TillichZemor2014,krishna2021fault}.
These features make HGP codes a natural setting for studying how code
structure constrains and enables the physical implementation of logical gates.

Logical entangling layers are indispensable to fault-tolerant Clifford
computation, logical state preparation, entanglement generation, and
measurement-based routines, and they can contribute substantially to the
physical spacetime overhead of these tasks. Their implementation commonly uses
transversal gates between code blocks, homological gadgets, code
deformation, or other architecture-dependent procedures
\cite{Fowler2012Surface,Horsman2012LatticeSurgery,cowtan2026parallel,
xu2025fast,ide2025fault}.
This overhead has motivated logical-gate constructions that exploit additional
structure already present in the code. Addressable logical gates, automorphism
gates, and structured physical permutations can realize nontrivial logical
operations through code symmetries
\cite{he2025asymptotically,sayginel2025fault,davydova2024quantum}.
Such constructions can substantially reduce the cost of logical gates, while
typically retaining active physical operations. 

Phantom codes use code symmetries to implement every ordered in-block logical
CNOT by a physical qubit permutation together with a Pauli-frame update. Ref.~\cite{Koh2026Phantom} introduced this concept,
constructed explicit families, and demonstrated its possible circuit-level
advantages. Other work has shown that this mechanism strongly constrains the
available code parameters. \citet{MorrisMalz2026PhantomBounds} derived strong upper bounds on the
number of logical qubits that can be encoded by binary phantom codes. \citet{guyot2026addressability} derived related
limitations on implementing logical gates by physical qubit permutations in
high-rate CSS codes. Together, these results show that implementing logical gates by physical qubit
permutations places strong constraints on qLDPC code design. They
motivate asking how much design freedom remains when the phantom condition is
imposed within a standard qLDPC construction.

\begin{figure}[t]
\centering
\includegraphics[width=\columnwidth]{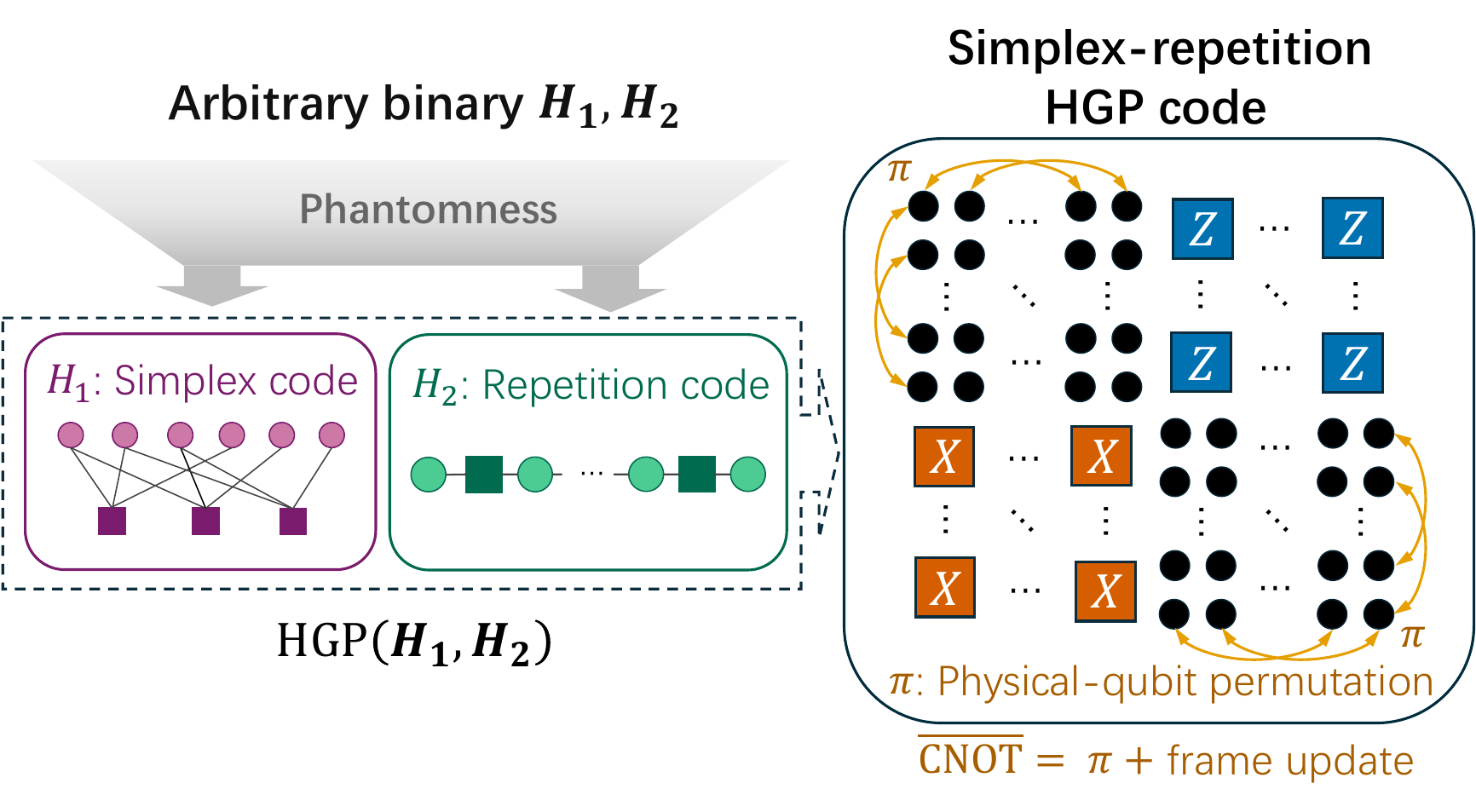}
\caption{
Schematic illustration of the simplex--repetition HGP phantom code.
The gray funnel represents the classification result. Up to natural
equivalences, the phantom condition selects the simplex--repetition family
among binary CSS HGP constructions. The two classical seed codes $H_1, H_2$ are the
simplex code and the repetition code. Their hypergraph product gives the corresponding
HGP phantom code. Black dots denote physical qubits, while orange and blue
squares denote \(X\)- and \(Z\)-checks, respectively. Orange arrows depict a
physical qubit permutation \(\pi\). For each ordered in-block logical CNOT, an appropriate permutation \(\pi\)
and a Pauli-frame update realize the logical operation.
}
\label{fig:intro_diagram}
\end{figure}

In this work, we give a complete classification of binary CSS HGP phantom
codes. We prove that, up to natural equivalences, the simplex--repetition
family is the unique binary CSS HGP family satisfying the phantom condition.
This result reduces a broad search over arbitrary HGP inputs to a single family
and complements existing constructions and general bounds for
permutation-based logical gates by providing an exact classification. We then
test the circuit-level performance of this family in logical GHZ preparation
and Trotterized many-body quantum simulation. The simplex--repetition HGP codes
retain low-weight stabilizer checks and yield concrete advantages over
distance-matched rotated surface-code baselines in both benchmarks because
their in-block logical CNOTs are implemented by physical qubit permutations
and Pauli-frame updates. This combination is well suited to neutral-atom architectures:
reconfigurable connectivity can support the nonlocal operations required by
qLDPC circuits, while the phantom in-block CNOTs are handled through classical
relabelling and Pauli-frame tracking. By reducing the full HGP design space to a single
family, our classification makes precise how permutation-based logical
entangling constrains code design within a standard qLDPC framework. It also
motivates the search for phantom code families with better asymptotic
parameters in broader qLDPC constructions.
Figure~\ref{fig:intro_diagram} gives a schematic view of this family and the
implementation of its logical CNOTs by physical qubit permutations.

\emph{Uniqueness of phantom codes in HGP codes.}---
We now classify binary CSS hypergraph product (HGP) codes satisfying the
phantom condition. Let
\(H_1\in\mathbb F_2^{r_1\times n_1}\) and
\(H_2\in\mathbb F_2^{r_2\times n_2}\) be two binary parity-check matrices, and
let
\(\mathcal C=\mathrm{HGP}(H_1,H_2)\)
denote the corresponding HGP code. Its check matrices are
\begin{equation}
\begin{aligned}
H_X&=
\begin{bmatrix}
H_1\otimes I_{n_2} & I_{r_1}\otimes H_2^{T}
\end{bmatrix},
\\
H_Z&=
\begin{bmatrix}
I_{n_1}\otimes H_2 & H_1^{T}\otimes I_{r_2}
\end{bmatrix}.
\end{aligned}
\label{eq:hgp_checks}
\end{equation}
For a binary matrix \(M\), \(\operatorname{row}(M)\) and \(\ker M\) denote
its row space and kernel, respectively. The \(X\)- and \(Z\)-stabilizer spaces are
\begin{equation}
S_X=\operatorname{row}(H_X),
\quad
S_Z=\operatorname{row}(H_Z).
\end{equation}
The \(X\)-logical and \(Z\)-logical operator spaces are
\begin{equation}
L_X=\ker H_Z/S_X,
\quad
L_Z=\ker H_X/S_Z.
\end{equation}

Following Ref.~\cite{Koh2026Phantom}, we recall the definition of a CSS
phantom code.
\begin{definition}[CSS phantom code~\cite{Koh2026Phantom}]
A \([[n,k,d]]\) CSS code is a CSS phantom code if the logical gate
\(\overline{\mathrm{CNOT}}_{ab}\) for every ordered pair of logical qubits
\((a,b)\in [k]^2\) can be implemented via qubit permutations, for some choice
of logical basis.
\end{definition}

A physical qubit permutation is called CSS-preserving if it maps \(X\)-type
stabilizers to \(X\)-type stabilizers and \(Z\)-type stabilizers to \(Z\)-type
stabilizers. Such a permutation induces a well-defined linear action on the
logical Pauli quotients. In this work, we consider the stronger condition that
the induced action is transitive on the nonzero logical classes.

The main result of this section is the following. Here and below, trivial transformations mean exchanging the two input codes,
passing to the transpose-side variant, and adding or removing zero or repeated
coordinates.

\begin{theorem}[Uniqueness of HGP phantom codes]
\label{thm:uniqueness_hgp_phantom}
Up to trivial transformations, the HGP phantom codes are exactly
the HGPs of a simplex code and a repetition code.
\end{theorem}

Theorem~\ref{thm:uniqueness_hgp_phantom} shows that, within binary CSS HGP
codes, requiring all ordered in-block logical CNOTs to be implemented by
physical qubit permutations singles out the simplex--repetition family.

The theorem below provides the structural input behind this main result.

\begin{theorem}[Transitive permutation actions in HGP codes]
\label{thm:hgp_transitive_action}
Let \(\mathcal C=\mathrm{HGP}(H_1,H_2)\). Let
\(\operatorname{Aut}_{\mathrm{perm}}(\mathcal C)\) be its permutation
automorphism group, and let
\begin{equation}
\rho_X:\operatorname{Aut}_{\mathrm{perm}}(\mathcal C)\to GL(L_X)
\end{equation}
be the induced action of qubit permutations on \(X\)-logical operators. If
\(\rho_X(\operatorname{Aut}_{\mathrm{perm}}(\mathcal C))\) acts transitively on
\(L_X\setminus\{0\}\), then, up to trivial transformations, the \(X\)-logical
sector is the one obtained from the HGP of a simplex code and a repetition
code. The same statement holds for \(Z\)-logical operators.
\end{theorem}

We first sketch the proof of Theorem~\ref{thm:hgp_transitive_action}, from
which Theorem~\ref{thm:uniqueness_hgp_phantom} follows immediately. The key quantity is the \textit{minimum coset
weight}. For \(x\in\ker H_Z\), let \([x]\in L_X\) denote its logical class, and
define
\begin{equation}
\label{eq:def-omega}
\omega([x])
\coloneqq
\min_{s\in S_X}\operatorname{wt}(x+s).
\end{equation}
Physical qubit permutations preserve Hamming weight and map logical cosets to
logical cosets, so \(\omega([x])\) is invariant under the induced action on
logical classes. Therefore, if the permutation action is transitive on
\(L_X\setminus\{0\}\), then all nonzero logical classes must have the same
minimum logical coset weight.

The proof of Theorem~\ref{thm:hgp_transitive_action} compares this constant
weight condition with the K\"unneth decomposition of the HGP logical space.
First, the two K\"unneth sectors cannot both be nonzero: if they were, then the
sum of a nonzero class from each sector would necessarily carry irreducible
contributions from both sides, forcing its minimum logical coset weight to be
strictly larger than the common value above. This gives a contradiction.
Hence only one K\"unneth sector can remain.

It then remains to analyze the active tensor sector. If both tensor factors had
dimension at least two, then rank-two tensors would produce nonzero logical
classes whose minimum representatives cannot stay in the same weight orbit as
the decomposable tensors. This gives a second contradiction. Therefore one
factor must be one-dimensional. After removing zero coordinates, repeated
coordinates, and redundant checks, the one-dimensional factor is the repetition
code, while the remaining factor is the simplex code. This proves
Theorem~\ref{thm:hgp_transitive_action}. The full argument is given in
Appendix~\ref{app:transitive_action_proof}.

We now explain Theorem~\ref{thm:uniqueness_hgp_phantom}. After choosing a
logical \(X\)-basis, logical CNOT gates act on \(L_X\) as the elementary
transvections. For \(k=1\) there is nothing to prove, while for \(k\ge 2\) the
elementary transvections act transitively on \(L_X\setminus\{0\}\).
Theorem~\ref{thm:hgp_transitive_action} therefore implies that, up to trivial
transformations, the HGP construction must be the HGP of a simplex code and a
repetition code. Conversely, the HGP of a simplex code and a repetition code is
phantom by \cite{Koh2026Phantom}. This proves
Theorem~\ref{thm:uniqueness_hgp_phantom}.

Theorem~\ref{thm:uniqueness_hgp_phantom} therefore places a sharp HGP
restriction on the search for phantom qLDPC codes.
Ref.~\cite{Koh2026Phantom} pointed out that phantom codes realize the
zero-overhead limit for logical entangling gates and suggested that this regime
may exclude high-rate qLDPC codes. Our classification turns this general
concern into an exact structural statement within binary CSS HGP codes.
Requiring every ordered in-block logical CNOT to be implemented by a physical
qubit permutation leaves only the HGP of a simplex code and a repetition code. This result complements the broader constraints established in
Refs.~\cite{MorrisMalz2026PhantomBounds,guyot2026addressability}.
Ref.~\cite{MorrisMalz2026PhantomBounds} proves the bound
\(k\le \log_2(n+1)\) for binary phantom codes with \(d\ge 2\) and \(k\neq 4\).
Ref.~\cite{guyot2026addressability} shows that high-rate CSS codes do not
generally support logical permutations through physical permutations and
derives a related no-go result for logical CNOT and CZ gates between two such
codes under a more restrictive circuit assumption. These works characterize
general rate and addressability constraints, while
Theorem~\ref{thm:uniqueness_hgp_phantom} classifies all binary CSS HGP codes
satisfying the phantom condition. A natural next step is to determine whether
phantom qLDPC families with genuinely scalable parameters can exist in broader
code constructions.


\emph{Canonical simplex--repetition HGP family.}---
Theorem~\ref{thm:uniqueness_hgp_phantom} identifies the
simplex--repetition HGP family as the unique family to study. We choose
\(H_1\) as a parity-check matrix for the binary simplex code. Its columns are
indexed by the nonzero vectors in \(\Ftwo^k\), and its weight-three rows are
supported on triples of the form \(\{u,v,u+v\}\). Hence
\begin{equation}
    n_1=2^k-1,
    \qquad
    m_1=\frac{(2^k-1)(2^{k-1}-1)}{3}.
\end{equation}
We choose \(H_2\) as the standard parity-check matrix of the length-\(d\)
repetition code, with each row supported on two neighboring coordinates. Thus
\begin{equation}
    n_2=d,
    \qquad
    m_2=d-1.
\end{equation}
This family already appears in
Ref.~\cite[Appendix~G.1]{Koh2026Phantom}; here it arises as the unique HGP
family selected by Theorem~\ref{thm:uniqueness_hgp_phantom}.

The resulting HGP code has parameters
\begin{equation}
    [[n_1n_2+m_1m_2,\,k,\,\min\{2^{k-1},n_2\}]].
\end{equation}
The main operational feature of this family is that every ordered in-block
logical CNOT can be implemented by a physical qubit permutation and a
Pauli-frame update while the stabilizer checks remain low weight. The
weight-three checks of the simplex code and the weight-two checks of the
repetition code give an \(X\)-check weight of at most \(5\) and a
\(Z\)-check weight of \(2^{k-1}+1\). The parameter \(k\) therefore controls
both the number of logical qubits encoded in each block and the complexity of
syndrome extraction and decoding. The construction details and the complete
row and column weight bounds are given in
Appendix~\ref{app:construction_details}.

For the circuit-level study, we focus on \(k=4\), which
provides a practical balance between the number of logical qubits per block
and the stabilizer-check and column weights. Each block encodes four logical
qubits, allowing the four-qubit GHZ seed to be prepared entirely within one
HGP block, while the \(X\)- and \(Z\)-check weights are \(5\) and \(9\). For \(k=4\), the code distance is
\(\min\{8,d\}\). We therefore choose \(d=8\), which reaches the full distance supported by the simplex code and permits a
distance-matched comparison with the distance-eight rotated surface-code
baseline. In this case,
\begin{equation}
    n_{\mathrm{HGP}}=15\cdot 8+35\cdot 7=365,
\end{equation}
so the primary HGP instance in the benchmark is the
\([[365,4,8]]\) code.


\emph{Circuit benchmarks}---
We now evaluate the simplex--repetition HGP family at the circuit level. For
\(k=4\), the family has parameters
\[
[[50d-35,4,d]],
\qquad d=3,\ldots,8.
\]
The low stabilizer-check weights and moderate column weights allow these codes
to be compiled directly into syndrome-extraction circuits. 

We generate HGP syndrome-extraction circuits using the cardinal construction
of Ref.~\cite{Tremblay2022ThinPlanar} together with the balanced sign
assignment of Ref.~\cite{Kang2025QUITS}. Each edge of the product Tanner graph
is assigned one of the four directions \(E,N,S,W\), and edge coloring
partitions the corresponding CNOTs into collision-free parallel layers. The
balanced assignment distributes incident edges between opposite directions to
reduce the circuit depth. This produces bare-ancilla syndrome-extraction
circuits and the corresponding \texttt{Stim} detector error models
\cite{Gidney2021Stim}. We use this pipeline to study logical GHZ-state preparation and Trotterized many-body quantum simulation.

Both benchmarks compare the \(k=4\) simplex--repetition HGP family with
rotated surface-code baselines under the circuit-level noise model of
Ref.~\cite{Koh2026Phantom}. In the HGP circuits, in-block logical CNOTs are
implemented by physical qubit permutations and Pauli-frame updates, while
logical CNOTs between different HGP blocks are implemented transversally. The
surface-code baseline uses one rotated surface-code block per logical qubit
and implements every logical CNOT by transversal physical CNOTs between
blocks. For each benchmark, the two implementations follow the same logical
gate schedule, with one QEC cycle performed after each transversal CNOT layer
between code blocks. Both use ideal logical initialization and noiseless final
logical readout.
The HGP detector error models are decoded with
BP+LSD~\cite{Hillmann2025LSD}. The surface-code circuits are decoded with the
logical-observable MWPM decoder implemented in the open-source
\texttt{lomatching} package~\cite{Serra_Peralta_2026}. The two benchmarks therefore test the benefit of the permutation-based
implementation in logical entanglement generation and in a circuit with
repeated many-body interactions.

For both benchmarks, let \(\mathcal S\) denote the task-specific set of
decoded final logical observables with known ideal outcomes. We define the
logical infidelity by
\begin{equation}
1-F
=
\frac{m}{N|\mathcal S|},
\label{eq:main_logical_infidelity}
\end{equation}
where \(N\) is the number of Monte Carlo shots and \(m\) is the total number
of decoded logical-observable outcomes that differ from their ideal values.
Thus the logical infidelity is the average mismatch probability per target
logical observable. It should be distinguished from a per-shot logical failure
probability. Further details are given in
Appendix~\ref{app:numerical_details}.

For logical GHZ preparation, the target state is
\begin{equation}
|\operatorname{GHZ}_K\rangle_L
=
\frac{|0\rangle_L^{\otimes K}+|1\rangle_L^{\otimes K}}{\sqrt2},
\qquad
K=8,16,32,64 .
\label{eq:main_ghz_state}
\end{equation}
Each HGP block encodes four logical qubits. We first prepare the four-qubit
in-block GHZ seed using phantom CNOTs. The state is then expanded across HGP
blocks using inter-block transversal CNOT layers arranged in a
logarithmic-depth binary-tree schedule.

Figure~\ref{fig:main_ghz} compares the HGP phantom construction with rotated
surface-code baselines at \(p=10^{-3}\) and \(p=5\times10^{-4}\) under the
circuit-level noise model and decoder choices considered here. The
distance-eight HGP code outperforms the distance-eight surface-code baseline
for all simulated GHZ sizes at both physical error rates. At \(p=10^{-3}\)
and \(K=64\), the logical infidelities are \(1.25\times10^{-4}\) for the HGP
code and \(7.14\times10^{-4}\) for the surface-code baseline. At
\(p=5\times10^{-4}\) and \(K=64\), the corresponding logical infidelities are
\(1.64\times10^{-5}\) and \(4.65\times10^{-5}\).

\begin{figure}[t]
\centering
\includegraphics[width=\linewidth]{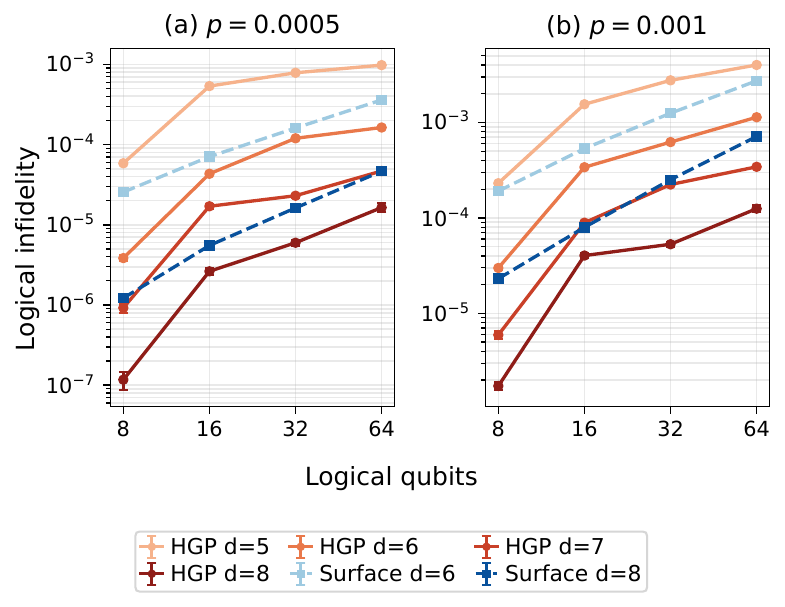}
\caption{
Logical GHZ-state preparation with the simplex--repetition HGP phantom family.
Left: \(p=5\times10^{-4}\). Right: \(p=10^{-3}\).
For each \(K\)-qubit GHZ target, HGP data use \(K/4\) four-logical-qubit
phantom blocks, whereas the surface-code baseline uses \(K\) rotated
surface-code blocks.  HGP distances are \(d=5,6,7,8\), and surface-code
baselines use \(d=6,8\).  The HGP implementation absorbs in-block CNOTs into
phantom relabellings, leaving only transversal inter-block CNOT layers followed
by QEC cycles.
}
\label{fig:main_ghz}
\end{figure}

For the Trotterized many-body benchmark, we follow
Ref.~\cite{Koh2026Phantom} and consider eight-body Ising interactions arranged
in an open-boundary brickwork pattern and interleaved with transverse-field
terms. At \(\theta=\phi=\pi/2\), the logical circuit is Clifford and can be
simulated at the stabilizer level. We retain the compiled representation rather
than algebraically simplifying the ideal unitary: each eight-body interaction
is implemented by a forward CNOT ladder, a logical \(Z\) rotation, and the
inverse CNOT ladder. This preserves the entangling structure of the benchmark
while enabling efficient stabilizer-level simulation. Figure~\ref{fig:main_trotter} shows the resulting comparison.


The distance-eight HGP code consistently improves over the distance-eight
surface-code baseline for all simulated \(K\) and \(p\).  At \(K=32\), the
logical infidelities are \(1.39\times10^{-2}\) versus \(3.17\times10^{-2}\) at
\(p=10^{-3}\), and \(2.93\times10^{-3}\) versus \(4.40\times10^{-3}\) at
\(p=5\times10^{-4}\), for HGP and surface-code implementations respectively.
The improvement is more modest than in GHZ preparation because the Trotterized
workload contains a denser inter-block entangling schedule: after the in-block
logical CNOTs are absorbed into phantom relabellings, many Trotter entangling
layers and intervening QEC cycles still remain physically active.  Even so, the data show that absorbing in-block logical CNOTs into
relabellings remains useful in a many-body circuit with active Trotter layers
and interleaved QEC.

\begin{figure}[t]
\centering
\includegraphics[width=\linewidth]{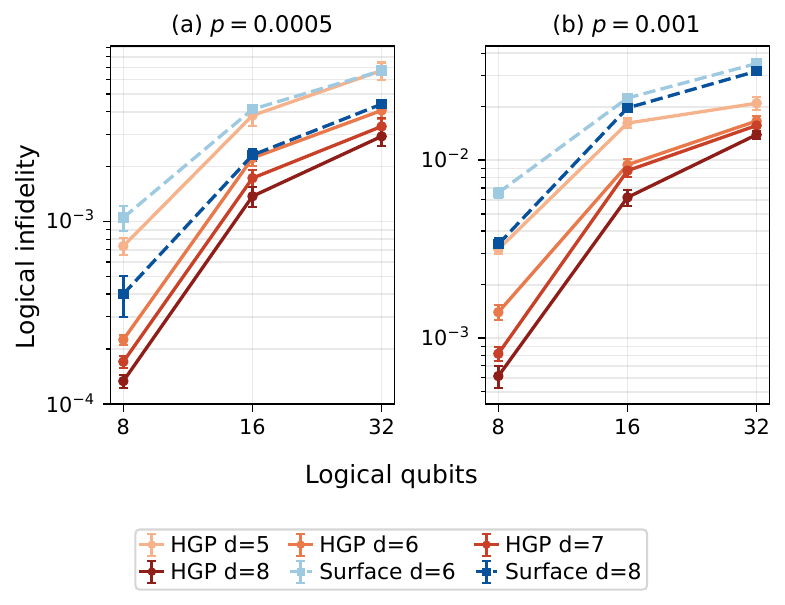}
\caption{
Trotterized many-body benchmark.
Left: \(p=5\times10^{-4}\). Right: \(p=10^{-3}\).
The HGP and surface-code configurations match those of the GHZ benchmark:
\(K/4\) four-logical-qubit HGP phantom blocks are compared with \(K\) rotated
surface-code blocks, using HGP distances \(d=5,6,7,8\) and surface-code
distances \(d=6,8\).
}
\label{fig:main_trotter}
\end{figure}

Together, the GHZ-preparation and Trotterized many-body benchmarks provide
circuit-level validation of the simplex--repetition HGP phantom family
identified by Theorem~\ref{thm:uniqueness_hgp_phantom}. In both tasks,
replacing in-block logical CNOT layers with physical qubit permutations and
Pauli-frame updates reduces the exposure to noisy entangling operations and
their associated QEC cycles. These circuit results complement the
classification by demonstrating the operational value of the
simplex--repetition family across logical state preparation and many-body
quantum simulation.

\emph{Conclusion and outlook}---
We have resolved, within binary CSS hypergraph product codes, a central
question about phantom logical entangling in fault-tolerant quantum
computation. Up to natural equivalences, the simplex--repetition family is the
unique HGP family satisfying the phantom condition. This classification shows
that requiring every ordered in-block logical CNOT to be implemented by a
physical qubit permutation imposes a strong restriction on qLDPC code design. 

We then studied this family at the circuit level through logical GHZ
preparation and Trotterized many-body quantum simulation. The
simplex--repetition family retains low-weight stabilizer checks and yields
concrete advantages over rotated surface-code baselines in both benchmarks
because its in-block logical CNOTs are implemented by physical qubit
permutations and Pauli-frame updates. More broadly, this mechanism suggests a
hardware--software co-design strategy in which selected logical entangling
operations are shifted from the quantum circuit to classical control through
relabelling of the data qubits. Such schemes may reduce two-qubit-gate depth,
routing and calibration overhead, and exposure to crosstalk and decoherence. Neutral-atom qLDPC architectures provide a promising setting for this approach,
with reconfigurable hardware supporting the remaining nonlocal operations and
phantom in-block CNOTs handled through classical relabelling.
The theorem and the two circuit benchmarks give a concrete HGP answer to the
broader question of combining phantom logical entangling with qLDPC-style
structure.

A natural next step is to determine whether phantom qLDPC families with better asymptotic parameters can arise from other standard constructions, such as balanced-product codes~\cite{breuckmann2021balanced} and lifted-product codes~\cite{panteleev2021quantum}. It would also be interesting to determine whether weak phantom codes, which support all permutation-based logical CNOTs on a selected subset of logical qubits, can form qLDPC families with scalable parameters. Such a relaxation may retain useful zero-overhead logical entangling operations while allowing greater freedom in code design, motivating the search for phantom qLDPC families with improved asymptotic parameters for low-overhead fault tolerance on neutral-atom hardware.

\textit{Acknowledgements}.---
This work was partially supported by the National Natural Science Foundation of China (Grant Nos.~92576114, 12447107) and the Guangdong Provincial Quantum Science Strategic Initiative (Grant Nos.~GDZX2403008, GDZX2503001, and GDZX2403001).

\bibliography{ref}


\vspace{4cm}
\onecolumngrid
\newpage

\appendix

\begin{center}
\Large{\textbf{Appendix --- Logical Entangling with Phantom Codes in Hypergraph Products} }
\end{center}

\appsection{app:transitive_action_proof}{Proof of Theorem~\ref{thm:hgp_transitive_action}}
\begin{proof}[Proof of Theorem~\ref{thm:hgp_transitive_action}]
The proof uses the K\"{u}nneth form of the HGP logical space. Since physical qubit permutations preserve minimum logical coset weight,
transitivity forces all nonzero logical classes to have the same minimum coset
weight. This weight constraint first eliminates one of
the two K\"{u}nneth terms, then forces the remaining term to have dimensions
\(1\times k\) or \(k\times1\). After zero and repeated coordinates are removed, the \(k\)-dimensional
classical component is a binary one-weight code and hence a simplex code.

Let $H_1\in\mathbb F_2^{r_1\times n_1},
\;
H_2\in\mathbb F_2^{r_2\times n_2}$
be two binary parity-check matrices. Denote the column and row index sets of \(H_i\) by $V_i=\{1,\ldots,n_i\},
\;
C_i=\{1,\ldots,r_i\}.$
The qubits of the hypergraph product code are divided into two disjoint sets,
\begin{equation}
Q_L=V_1\times V_2,
\qquad
Q_R=C_1\times C_2.
\end{equation}
Thus an \(X\)-type Pauli operator is represented by a pair of binary matrices $(X_L,X_R)
\in
\mathbb F_2^{n_1\times n_2}
\oplus
\mathbb F_2^{r_1\times r_2},$
where \(X_L\) is indexed by \(V_1\times V_2\), and \(X_R\) is indexed by \(C_1\times C_2\).

With this convention, the standard \(X\)-check matrix of the corresponding HGP code is
\begin{equation}
\label{eq:hx-standard}
H_X
=
\begin{bmatrix}
H_1\otimes I_{n_2}
&
I_{r_1}\otimes H_2^T
\end{bmatrix}.
\end{equation}
Its rows are indexed by \(C_1\times V_2\). Also denote $H_Z$ as the \(Z\)-check matrix of the corresponding HGP code. 

We systematically identify vectors indexed by Cartesian products with matrices. Let
\begin{equation}
\label{eq:rvec-def}
\operatorname{rvec}(X)\coloneqq \operatorname{vec}(X^T)
\end{equation}
denote row-wise vectorization. For compatible matrices \(A,X,B\), the vectorization gives
\begin{equation}
\label{eq:rvec-kron}
(A\otimes B)\operatorname{rvec}(X)
=
\operatorname{rvec}(AXB^T).
\end{equation}
In particular, for vectors \(u\in\mathbb F_2^m\) and \(v\in\mathbb F_2^n\),
\begin{equation}
\label{eq:rvec-rank-one}
u\otimes v
=
\operatorname{rvec}(uv^T).
\end{equation}
Hence, throughout the proof, we freely pass between tensor-product vectors and their reshaped matrix forms.

We first apply this convention to the row space of \(H_X\). Any binary linear combination of the rows of \(H_X\) represents an
\(X\)-type stabilizer. Write its coefficient vector as
\(\operatorname{rvec}(\Lambda)\), where
\(\Lambda\in\mathbb F_2^{r_1\times n_2}\) is indexed by
\(C_1\times V_2\). Although the rows of \(H_X\) are written as row vectors, it
is convenient to compute the transpose of their linear combination using
Eq.~\eqref{eq:rvec-kron}. The left component is
\begin{equation}
\begin{aligned}
(H_1\otimes I_{n_2})^T\operatorname{rvec}(\Lambda)
=
(H_1^T\otimes I_{n_2})\operatorname{rvec}(\Lambda) =
\operatorname{rvec}(H_1^T\Lambda),
\end{aligned}
\end{equation}
and the right component is
\begin{equation}
\begin{aligned}
(I_{r_1}\otimes H_2^T)^T\operatorname{rvec}(\Lambda)
=
(I_{r_1}\otimes H_2)\operatorname{rvec}(\Lambda) =
\operatorname{rvec}(\Lambda H_2^T).
\end{aligned}
\end{equation}
Therefore, after reshaping the vectorized components back into matrices, the \(X\)-stabilizer space is
\begin{equation}
\label{eq:sx-reshaped}
S_X
=
\operatorname{row}(H_X)
=
\left\{
\left(H_1^T\Lambda,\;\Lambda H_2^T\right):
\Lambda\in\mathbb F_2^{r_1\times n_2}
\right\}.
\end{equation}

The \(X\)-logical space has the K\"{u}nneth decomposition
\begin{equation}
\label{eq:Kunneth-lx}
L_X
\cong
\left(
\left(\mathbb F_2^{n_1}/\operatorname{im}H_1^T\right)
\otimes
\ker H_2
\right)
\oplus
\left(
\ker H_1^T
\otimes
\left(\mathbb F_2^{r_2}/\operatorname{im}H_2\right)
\right),
\end{equation}
with corresponding logical class denoted as
\begin{equation}
\begin{aligned}
    &[a]\otimes b,
\qquad
[a]\in \mathbb F_2^{n_1}/\operatorname{im}H_1^T,
\quad
b\in\ker H_2,\\
&c\otimes[d],
\qquad
c\in\ker H_1^T,
\quad
[d]\in\mathbb F_2^{r_2}/\operatorname{im}H_2.
\end{aligned}
\end{equation}
The \(X\)-logical space is also given by
\begin{equation}
L_X=\ker H_Z/S_X .
\end{equation}
For \(x\in\ker H_Z\), write
\begin{equation}
[x]=x+S_X\in L_X .
\end{equation}
Define \textit{the minimum coset weight} of \([x]\) by
\begin{equation}
\label{eq:app-def-omega}
\omega([x])
\coloneqq
\min_{s\in S_X}\operatorname{wt}(x+s).
\end{equation}

Next, we will show that the weight is well-defined on equivalence classes under permutation automorphism group $\operatorname{Aut}_{\mathrm{perm}}(\mathcal C)$. If $x'=x+s_0$
for some \(s_0\in S_X\), then $[x']=[x],$
and
\begin{equation}
\begin{aligned}
\omega([x'])
=
\min_{s\in S_X}\operatorname{wt}(x+s_0+s) =
\min_{u\in S_X}\operatorname{wt}(x+u) =
\omega([x]).
\end{aligned}
\end{equation}
Now let \(\pi\in\operatorname{Aut}_{\mathrm{perm}}(\mathcal C)\). Since \(\pi(S_X)=S_X\), the induced action on \(L_X\) is
\begin{equation}
\rho_X(\pi)([x])=[\pi (x)].
\end{equation}
Therefore
\begin{equation}
\begin{aligned}
\omega(\rho_X(\pi)([x]))
&=
\omega([\pi (x)]) \\
&=
\min_{t\in S_X}\operatorname{wt}(\pi (x)+t) \\
&=
\min_{s\in S_X}\operatorname{wt}(\pi (x)+\pi (s)) \\
&=
\min_{s\in S_X}\operatorname{wt}(\pi(x+s)) \\
&=
\min_{s\in S_X}\operatorname{wt}(x+s) \\
&=
\omega([x]).
\end{aligned}
\end{equation}
The fifth equality uses that a physical qubit permutation preserves Hamming weight. Since the group $\rho_X\bigl(\operatorname{Aut}_{\mathrm{perm}}(\mathcal C)\bigr)$
acts transitively on \(L_X\setminus\{0\}\), the quantity \(\omega([x])\) is constant on all nonzero logical classes. We denote this common value by
\begin{equation}\label{eq:w-constant}
    \omega([x])=w_0 >0 , \quad \forall 0\neq [x] \in L_X. 
\end{equation}

We next compute the minimum coset weight of decomposable logical classes in
each K\"{u}nneth summand. Consider $[a]\in \mathbb F_2^{n_1}/\operatorname{im}H_1^T,
\;
b\in\ker H_2,$ where \(a\in\mathbb F_2^{n_1}\) is a representative of the quotient class \([a]\). Define \textit{the minimum coset weight among all representatives} of the quotient class \([a]\)
\begin{equation}
\label{eq:delta-a-def}
\delta([a])
=
\min_{\gamma\in\mathbb F_2^{r_1}}
\operatorname{wt}(a+H_1^T\gamma).
\end{equation}
Using the reshaped matrix convention in Eq.~\eqref{eq:rvec-rank-one}, the tensor-product vector \(a\otimes b\), indexed by \(V_1\times V_2\), is written as the rank-one matrix $ab^T\in\mathbb F_2^{n_1\times n_2}.$
Thus the K\"{u}nneth summand \([a]\otimes b\) is represented on \(Q_L\) by \(ab^T\). It has no support on \(Q_R\), so its logical class in \(L_X\) is represented by
\begin{equation}
[(ab^T,0)].
\end{equation}
We claim that the minimum coset weight satisfies $\omega([(ab^T,0)])
=
\delta([a])\operatorname{wt}(b).$
After adding an arbitrary \(X\)-stabilizer, a representative of the same logical class has the form
\begin{equation}
\label{eq:first-class-representative}
\left(
ab^T+H_1^T\Lambda,\;
\Lambda H_2^T
\right),
\qquad
\Lambda\in\mathbb F_2^{r_1\times n_2}.
\end{equation}
Let \(\Lambda_j\in\mathbb F_2^{r_1}\) be the \(j\)-th column of \(\Lambda\). The \(j\)-th column of the left-sector matrix is
$b_ja+H_1^T\Lambda_j.$
If \(b_j=1\), this column is a representative of the quotient class \([a]\), and hence has weight at least \(\delta([a])\). There are exactly \(\operatorname{wt}(b)\) indices \(j\) with \(b_j=1\), and different columns occupy disjoint physical qubits in \(Q_L\). Therefore
\begin{equation}
\operatorname{wt}(ab^T+H_1^T\Lambda)
\ge
\delta([a])\operatorname{wt}(b).
\end{equation}
Since \(Q_L\) and \(Q_R\) are disjoint, we have
\begin{equation}
\begin{aligned}
\operatorname{wt}
\left(
ab^T+H_1^T\Lambda,\;
\Lambda H_2^T
\right)
&=
\operatorname{wt}(ab^T+H_1^T\Lambda)
+
\operatorname{wt}(\Lambda H_2^T) \\
&\ge
\delta([a])\operatorname{wt}(b)
+
\operatorname{wt}(\Lambda H_2^T) \\
&\ge
\delta([a])\operatorname{wt}(b).
\end{aligned}
\end{equation}
Since this holds for every stabilizer coefficient matrix \(\Lambda\), taking the minimum over all representatives gives
\begin{equation}
\label{eq:first-class-lower-bound}
\omega([(ab^T,0)])
\ge
\delta([a])\operatorname{wt}(b).
\end{equation}

Conversely, choose \(\gamma\in\mathbb F_2^{r_1}\) attaining the minimum in
Eq.~\eqref{eq:delta-a-def}, so that
\begin{equation}
\operatorname{wt}(a+H_1^T\gamma)
=
\delta([a]).
\end{equation}
Choose the particular \(X\)-stabilizer coefficient matrix satisfying
\begin{equation}
\Lambda_{\mathrm{min}}=\gamma b^T\in\mathbb F_2^{r_1\times n_2}.
\end{equation}
Then
\begin{equation}
H_1^T\Lambda_{\mathrm{min}}
=
H_1^T\gamma\, b^T,
\quad
\Lambda_{\mathrm{min}}H_2^T
=
\gamma b^T H_2^T
=
\gamma(H_2b)^T
=
0,
\end{equation}
for \(b\in\ker H_2\). Therefore adding the stabilizer corresponding to \(\Lambda_{\mathrm{min}}\) gives the representative
\begin{equation}
\left(
(a+H_1^T\gamma)b^T,\;
0
\right).
\end{equation}
The weight of this representative is
\begin{equation}
\begin{aligned}
\operatorname{wt}\left((a+H_1^T\gamma)b^T\right)
=
\operatorname{wt}(a+H_1^T\gamma)\operatorname{wt}(b) =
\delta([a])\operatorname{wt}(b).
\end{aligned}
\end{equation}
Thus
\begin{equation}
\omega([(ab^T,0)])
\le
\delta([a])\operatorname{wt}(b).
\end{equation}
Combining this with the lower bound in Eq.~\eqref{eq:first-class-lower-bound} gives
\begin{equation}
\label{eq:first-class-weight}
\omega([(ab^T,0)])
=
\delta([a])\operatorname{wt}(b).
\end{equation}

Similarly, consider $c\in\ker H_1^T,
\;
[d]\in\mathbb F_2^{r_2}/\operatorname{im}H_2,$
where \(d\in\mathbb F_2^{r_2}\) is a representative of \([d]\). Define
\begin{equation}
\label{eq:delta-d-def}
\delta([d])
=
\min_{\eta\in\mathbb F_2^{n_2}}
\operatorname{wt}(d+H_2\eta).
\end{equation}
The same argument gives
\begin{equation}
\label{eq:second-class-weight}
\omega([(0,cd^T)])
=
\operatorname{wt}(c)\delta([d]).
\end{equation}

While we have calculated the minimum coset weight of each K\"{u}nneth summand, we now show that they cannot both be nonzero, i.e., one of them must be zero. Suppose, for contradiction, that both types occur:
\begin{equation}
\begin{gathered}
0\neq [a]\in\mathbb F_2^{n_1}/\operatorname{im}H_1^T,
\qquad
0\neq b\in\ker H_2,\\
0\neq c\in\ker H_1^T,
\qquad
0\neq [d]\in\mathbb F_2^{r_2}/\operatorname{im}H_2,
\end{gathered}
\end{equation}
with representatives \(a\in\mathbb F_2^{n_1}\) and \(d\in\mathbb F_2^{r_2}\). The corresponding nonzero logical classes are $[(ab^T,0)]
$ and $
[(0,cd^T)].$ By Eq.~\eqref{eq:w-constant},  for all
nonzero $X$-logical classes, we have
\begin{equation}
\omega([(ab^T,0)])
=
\omega([(0,cd^T)])
=
w_0.
\end{equation}
Now consider the logical class obtained by adding these two contributions,
\begin{equation}
[(ab^T,cd^T)].
\end{equation}
It is nonzero because the two contributions lie in different sectors. After adding an arbitrary \(X\)-stabilizer, a representative of this class has the form
\begin{equation}
\left(
ab^T+H_1^T\Lambda,\;
cd^T+\Lambda H_2^T
\right).
\end{equation}
The lower-bound parts of the proofs of
Eqs.~\eqref{eq:first-class-weight} and~\eqref{eq:second-class-weight}
imply that, for every \(\Lambda\),
\begin{equation}
\begin{aligned}
\operatorname{wt}(ab^T+H_1^T\Lambda) & \geq \delta([a])\operatorname{wt}(b) = w_0, \\
\operatorname{wt}(cd^T+\Lambda H_2^T) & \geq \operatorname{wt}(c)\delta([d]) = w_0. 
\end{aligned}
\end{equation}
Since \(Q_L\) and \(Q_R\) are disjoint, every representative of \([(ab^T,cd^T)]\) has weight at least \(2w_0\). Hence
\begin{equation}
\omega([(ab^T,cd^T)])
\ge
2w_0.
\end{equation}
But \([(ab^T,cd^T)]\neq0\), so Eq.~\eqref{eq:w-constant} requires
\begin{equation}
\omega([(ab^T,cd^T)])
=
w_0.
\end{equation}
This is a contradiction. Therefore at most one K\"{u}nneth summand is nonzero.

Next we will show that the unique active tensor sector cannot have both factors of dimension greater than one. Without loss of generality, suppose that the first K\"{u}nneth summand is the
active one, with decomposable classes represented by
\begin{equation}
[(ab^T,0)],
\qquad
[a]\in\mathbb F_2^{n_1}/\operatorname{im}H_1^T,
\quad
b\in\ker H_2.
\end{equation}
Assume, for contradiction, that
\begin{equation}
\label{eq:both-factors-large-assumption}
\dim\left(\mathbb F_2^{n_1}/\operatorname{im}H_1^T\right)>1,
\qquad
\dim\ker H_2>1.
\end{equation}

For every nonzero \([a]\) and every nonzero \(b\), the class
\([(ab^T,0)]\) is nonzero. 
Fixing any nonzero \(b_0\in\ker H_2\), we get
$\delta([a])=\frac{w_0}{\operatorname{wt}(b_0)}$
for every nonzero \([a]\). Thus \(\delta([a])\) is constant on nonzero quotient classes. Similarly, fixing any nonzero \([a_0]\), we get that \(\operatorname{wt}(b)\) is constant on nonzero \(b\in\ker H_2\). 
Therefore there exist constants \(\alpha,\beta>0\) such that
\begin{equation}
\label{eq:alpha-beta-def}
\begin{aligned}
\delta([a])=\alpha
\qquad
& \forall\,0\neq [a]\in\mathbb F_2^{n_1}/\operatorname{im}H_1^T,\\
\operatorname{wt}(b)=\beta
\qquad
& \forall\,0\neq b\in\ker H_2.
\end{aligned}
\end{equation}
Then
\begin{equation}
\label{eq:D0-alpha-beta}
\omega([(ab^T,0)])
=
\delta([a])\operatorname{wt}(b) = \alpha\beta = w_0.
\end{equation}

Choose linearly independent quotient classes
\([a],[a']\in\mathbb F_2^{n_1}/\operatorname{im}H_1^T\) and linearly
independent vectors \(b,b'\in\ker H_2\). Consider
\begin{equation}
[(ab^T+a'(b')^T,0)].
\end{equation}
This class is nonzero because the tensor $[a]\otimes b+[a']\otimes b'$
is nonzero in the active K\"{u}nneth summand. After adding an arbitrary \(X\)-stabilizer, a representative has the form
\begin{equation}
(ab^T+a'(b')^T+H_1^T\Lambda, \Lambda H_2^T).
\end{equation}
The \(j\)-th column of the left part is
\begin{equation}
b_ja+b'_ja'+H_1^T\Lambda_j.
\end{equation}
Modulo \(\operatorname{im}H_1^T\), this column represents
\begin{equation}
b_j[a]+b'_j[a']\in
\mathbb F_2^{n_1}/\operatorname{im}H_1^T.
\end{equation}
Whenever \((b_j,b'_j)\neq(0,0)\), this quotient class is nonzero because \([a]\) and \([a']\) are linearly independent. Hence the \(j\)-th column has weight at least \(\alpha\). Therefore
\begin{equation}
\begin{aligned}
    \operatorname{wt}(ab^T+a'(b')^T+H_1^T\Lambda, \Lambda H_2^T) &= \operatorname{wt}(ab^T+a'(b')^T+H_1^T\Lambda) + \operatorname{wt}(\Lambda H_2^T)  \\
    &\ge \alpha
\left|
\operatorname{supp}(b)\cup\operatorname{supp}(b')
\right|.
\end{aligned}
\end{equation}
Thus
\begin{equation}
\label{eq:rank-two-lower-union}
\omega([(ab^T+a'(b')^T,0)])
\ge
\alpha
\left|
\operatorname{supp}(b)\cup\operatorname{supp}(b')
\right|.
\end{equation}

Since \(b\), \(b'\), and \(b+b'\) are nonzero elements of \(\ker H_2\), and every nonzero element of \(\ker H_2\) has weight \(\beta\), we have
\begin{equation}
\operatorname{wt}(b)
=
\operatorname{wt}(b')
=
\operatorname{wt}(b+b')
=
\beta.
\end{equation}
In particular, this also shows that 
 $\beta$ is even. 
For binary vectors,
\begin{equation}
\operatorname{wt}(b+b')
=
\operatorname{wt}(b)+\operatorname{wt}(b')
-
2
\left|
\operatorname{supp}(b)\cap\operatorname{supp}(b')
\right|.
\end{equation}
Therefore
\begin{equation}
\left|
\operatorname{supp}(b)\cap\operatorname{supp}(b')
\right|
=
\frac{\beta}{2},
\end{equation}
and hence
\begin{equation}
\left|
\operatorname{supp}(b)\cup\operatorname{supp}(b')
\right|
=
2\beta-\frac{\beta}{2}
=
\frac{3\beta}{2}.
\end{equation}
Using Eq.~\eqref{eq:rank-two-lower-union}, we get
\begin{equation}
\omega([(ab^T+a'(b')^T,0)])
\ge
\alpha\frac{3\beta}{2}
=
\frac{3}{2}w_0
>
w_0.
\end{equation}
But \([(ab^T+a'(b')^T,0)]\neq0\), so Eq.~\eqref{eq:w-constant} requires
\begin{equation}
\omega([(ab^T+a'(b')^T,0)])
=
w_0.
\end{equation}
This is a contradiction. Therefore the unique active tensor sector must have one factor of dimension one:
\begin{equation}
\label{eq:one-factor-one-dimensional}
\dim\left(\mathbb F_2^{n_1}/\operatorname{im}H_1^T\right)=1, \;\dim\ker H_2=k 
\quad\text{or}\quad
\dim\ker H_2=1, \; \dim\left(\mathbb F_2^{n_1}/\operatorname{im}H_1^T\right)=k.
\end{equation}
The same argument applies when the second K\"{u}nneth summand 
\([(0,cd^T)]\) is active. 

Finally, we identify the classical codes in the remaining case. We now show
that, after zero coordinates and repeated coordinate copies are removed, the
construction is precisely the HGP construction of a repetition code and a
simplex code. Assume first that the active K\"{u}nneth summand is represented by
\begin{equation}
[(ab^T,0)],
\qquad
[a]\in\mathbb F_2^{n_1}/\operatorname{im}H_1^T,
\quad
b\in\ker H_2,
\end{equation}
with
\begin{equation}
\label{eq:representative-final-case}
\dim\left(\mathbb F_2^{n_1}/\operatorname{im}H_1^T\right)=1,
\qquad
\dim\ker H_2=k.
\end{equation}
Equation~\eqref{eq:representative-final-case} and rank--nullity give
\begin{equation}
\dim\left(\mathbb F_2^{n_1}/\operatorname{im}H_1^T\right)
=
n_1-\operatorname{rank}H_1
=
\dim\ker H_1
=
1.
\end{equation}
After deleting coordinates that are identically zero on \(\ker H_1\), the unique
nonzero codeword of \(\ker H_1\) has full support. Thus, up to a coordinate
permutation, \(\ker H_1\) is generated by the all-one vector and is a repetition
code.

We next identify the \(k\)-dimensional code. Let $0\neq [a_0]\in
\mathbb F_2^{n_1}/\operatorname{im}H_1^T$
be the unique nonzero class, and choose a representative
\(a_0\in\mathbb F_2^{n_1}\). For every nonzero \(b\in\ker H_2\), the logical
class \([(a_0b^T,0)]\) is nonzero. Hence Eqs.~\eqref{eq:first-class-weight}
and~\eqref{eq:w-constant} imply
\begin{equation}
   \omega([(a_0b^T,0)])
=
\delta([a_0])\operatorname{wt}(b)
=
w_0 . 
\end{equation}
Since \(\delta([a_0])>0\) is fixed, all nonzero codewords of
\(\ker H_2\) have the same Hamming weight. Hence \(\ker H_2\) is a binary
linear one-weight code. By the classification of linear one-weight codes~\cite{Bonisoli1984Equidistant,Kiermaier_2023}, it is
permutation-equivalent to a replication of the binary simplex code, possibly
with additional zero coordinates. Therefore, after deleting zero coordinates
and retaining one representative from each repeated coordinate class,
\(\ker H_2\) is the binary simplex code
\([2^k-1,k,2^{k-1}].\)

The case with dimensions \(\dim\ker H_2=1, \; \dim\left(\mathbb F_2^{n_1}/\operatorname{im}H_1^T\right)=k\), as well as the cases represented by
\([(0,cd^T)]\), are obtained by exchanging the two input matrices or by passing
to the transpose-side representation. Thus, apart from zero coordinates and repeated coordinate copies, the only
possibility compatible with the transitive logical permutation action is the HGP
construction built from a simplex code and a repetition code. This completes the proof.

\end{proof}


\appsection{app:construction_details}{The simplex--repetition HGP family}

\appsubsection{app:simplex_repetition_construction}
{The simplex--repetition HGP construction}

We give a description of the simplex--repetition construction
used in the main text, following
Ref.~\cite[Appendix~G.1]{Koh2026Phantom}. For distinct
\(u,v\in\Ftwo^k\setminus\{0\}\), define
\begin{equation}
    \ell(u,v)=\{u,v,u+v\}.
\end{equation}
Over \(\Ftwo\), each such set contains three distinct nonzero vectors. We
define \(H_1\) to have columns indexed by the nonzero vectors in \(\Ftwo^k\)
and rows indexed by the distinct sets \(\ell(u,v)\), with
\begin{equation}
    (H_1)_{\ell,p}
    =
    \begin{cases}
    1, & p\in\ell,\\
    0, & p\notin\ell .
    \end{cases}
\end{equation}
There are $n_1=2^k-1$
columns. Each line is generated by six ordered pairs of distinct nonzero
vectors, giving
\begin{equation}
    m_1
    =
    \frac{(2^k-1)(2^k-2)}{6}
    =
    \frac{(2^k-1)(2^{k-1}-1)}{3}.
\end{equation}
Every row of \(H_1\) has weight three, while every column has weight
\(2^{k-1}-1\). Although this presentation contains redundant parity checks,
it is useful for the HGP construction because it combines weight-three rows
with a permutation action of the full group \(\GL(k,\Ftwo)\).

A vector $f:\Ftwo^k\setminus\{0\}\longrightarrow\Ftwo$
lies in \(\ker H_1\) precisely when
\begin{equation}
    f(u)+f(v)+f(u+v)=0
\end{equation}
for every pair of distinct nonzero vectors \(u,v\). Extending \(f\) by
\(f(0)=0\), this condition becomes
\begin{equation}
    f(u+v)=f(u)+f(v)
\end{equation}
for all \(u,v\in\Ftwo^k\). Hence \(f\) is a linear functional on
\(\Ftwo^k\). It follows that \(\ker H_1\) is the binary simplex code with
parameters
\begin{equation}
    [2^k-1,\,k,\,2^{k-1}].
\end{equation}

For the second classical code, we use the standard parity-check matrix of the
length-\(n_2\) repetition code,
\begin{equation}
    (H_2)_{i,j}
    =
    \begin{cases}
    1, & j\in\{i,i+1\},\\
    0, & \text{otherwise},
    \end{cases}
    \qquad
    1\le i\le n_2-1 .
\end{equation}
Thus $m_2=n_2-1.$
Every row of \(H_2\) has weight two, and
\(\ker H_2\) is the repetition code with parameters
\([n_2,1,n_2]\). Since \(H_2\) has full row rank, the transpose-side logical
contribution in the HGP construction vanishes. For this pair of classical
codes, the HGP construction has parameters
\begin{equation}
    [[n_1n_2+m_1m_2,\,k,\,\min\{2^{k-1},n_2\}]].
\end{equation}
The logical dimension is inherited from the \(k\)-dimensional simplex code,
while the distance is limited by the smaller of the simplex-code distance and
the repetition-code length.

We next describe the physical permutations responsible for the phantom
logical CNOTs. For any \(A\in\GL(k,\Ftwo)\), the map
\begin{equation}
    p\longmapsto Ap
\end{equation}
permutes the nonzero vectors in \(\Ftwo^k\). It also permutes the weight-three
checks because
\begin{equation}
    \{u,v,u+v\}
    \longmapsto
    \{Au,Av,A(u+v)\}.
\end{equation}
Let \(\sigma_A\) denote the induced permutation of the columns of \(H_1\), and
let \(\pi_A\) denote the corresponding permutation of its rows. Their
compatibility is expressed by
\begin{equation}
    H_1\sigma_A=\pi_AH_1.
\end{equation}

To lift this symmetry of \(H_1\) to a physical qubit permutation of the HGP
code, recall that the HGP qubits are indexed by the disjoint union
\begin{equation}
    (V_1\times V_2)\sqcup(C_1\times C_2),
\end{equation}
where \(V_1\) and \(C_1\) are the column and row index sets of \(H_1\), and
\(V_2\) and \(C_2\) are those of \(H_2\). Thus, a qubit of the first type is
labelled by \((p,j)\), where \(p\) is a nonzero vector in \(\Ftwo^k\) and
\(j\) is a coordinate of the repetition code. A qubit of the second type is
labelled by \((\ell,i)\), where \(\ell\) is a simplex-code check and \(i\) is
a repetition-code check.

Since \(A\) permutes both the points and the projective lines associated with
\(H_1\), it induces the following permutation of the HGP qubits:
\begin{equation}
    (p,j)\longmapsto(Ap,j),
    \qquad
    (\ell,i)\longmapsto(A\ell,i),
\end{equation}
where
\begin{equation}
    A\ell=\{Au,Av,A(u+v)\}
\end{equation}
for \(\ell=\{u,v,u+v\}\). The transformation acts only on the
simplex-code indices and leaves the repetition-code indices unchanged.

We now verify that this qubit permutation preserves the CSS stabilizer group.
Under this permutation, the columns of
\begin{equation}
    H_X=
    [\,H_1\otimes I_{n_2}\mid I_{m_1}\otimes H_2^{T}\,]
\end{equation}
indexed by \(V_1\times V_2\) are permuted by \(\sigma_A\) on the \(V_1\)
coordinate, while those indexed by \(C_1\times C_2\) are permuted by
\(\pi_A\) on the \(C_1\) coordinate. The relation
\(H_1\sigma_A=\pi_AH_1\) shows that the transformed matrix differs from
\(H_X\) only by a permutation of its rows on the \(C_1\) coordinate.
Therefore its row space is preserved.

Similarly, for
\begin{equation}
    H_Z=
    [\,I_{n_1}\otimes H_2\mid H_1^{T}\otimes I_{m_2}\,],
\end{equation}
the transformed matrix differs from \(H_Z\) only by a permutation of its rows
on the \(V_1\) coordinate. Hence the row spaces of both \(H_X\) and \(H_Z\)
are preserved, and every \(A\in\GL(k,\Ftwo)\) induces a CSS-preserving
physical qubit permutation of the HGP code.

Choose a logical basis and label the logical Pauli operators by
\(a,b\in\Ftwo^k\), writing them as \(X(a)\) and \(Z(b)\). The induced logical
action is
\begin{equation}
    X(a)\longmapsto X(Aa),
    \qquad
    Z(b)\longmapsto Z(A^{-T}b).
\end{equation}
For any ordered pair \(i\neq j\), choose the elementary transvection
\begin{equation}
    A=I+e_je_i^{T}.
\end{equation}
This transformation sends \(e_i\) to \(e_i+e_j\) and fixes every
\(e_\ell\) with \(\ell\neq i\). It therefore induces the logical action of
\(\overline{\mathrm{CNOT}}_{i j}\). Varying \(i\) and \(j\) realizes every
ordered in-block logical CNOT through a physical qubit permutation, with the
logical Pauli frame updated accordingly. This establishes the phantom property
of the simplex--repetition HGP family.

\appsubsection{app:check_weight_bounds}{Check-weight bounds and representative family members}

For a CSS code, the maximum row weights determine the largest stabilizer
checks and the number of two-qubit gates required for their direct
syndrome measurement. The maximum column weights count how many
stabilizer checks are incident on a physical qubit and therefore influence
circuit scheduling and the complexity of the corresponding decoding problem.
We now evaluate these quantities for the simplex--repetition HGP family.

For a binary matrix \(M\), let \(\max\operatorname{rowwt}(M)\) and
\(\max\operatorname{colwt}(M)\) denote the maximum row weight and maximum
column weight of \(M\), respectively.

\begin{theorem}[Row and column weight bounds]
\label{thm:row_col_weight_bounds_polished}
Let \(H_1\) be the simplex parity-check matrix whose weight-three rows are
supported on triples \(\{u,v,u+v\}\), and let \(H_2\) be the standard
nearest-neighbor parity-check matrix of the repetition code. Then the
\(X\)- and \(Z\)-check matrices of the simplex--repetition HGP construction
satisfy
\begin{equation}
\max\operatorname{rowwt}(H_X)=
\begin{cases}
4, & n_2=2,\\
5, & n_2\ge 3,
\end{cases}
\end{equation}
\begin{equation}
\max\operatorname{colwt}(H_X)=
\max\left\{\frac{n_1-1}{2},\,2\right\},
\end{equation}
\begin{equation}
\max\operatorname{rowwt}(H_Z)=\frac{n_1+3}{2},
\qquad
\max\operatorname{colwt}(H_Z)=3.
\end{equation}
\end{theorem}

\begin{proof}
For
\[
H_X=
[\,H_1\otimes I_{n_2}\mid I_{m_1}\otimes H_2^{T}\,],
\]
each row contains the three qubits supported by a simplex check together with
the qubits incident on the corresponding coordinate of the repetition code.
An endpoint coordinate contributes one additional qubit, while an interior
coordinate contributes two. The maximum \(X\)-check weight is therefore \(5\)
for \(n_2\ge3\). When \(n_2=2\), both repetition-code coordinates are
endpoints, giving a maximum weight of \(4\). A column in the first block of \(H_X\) has weight equal to the number of
simplex checks containing a fixed coordinate. Every nonzero vector in
\(\Ftwo^k\) belongs to
\((n_1-1)/2\) triples of the form \(\{u,v,u+v\}\). A column in the second block has weight
\(2\), inherited from the repetition-code parity checks. This gives \(\max\operatorname{colwt}(H_X)=
\max\left\{\frac{n_1-1}{2},\,2\right\}.\)

For
\[
H_Z=
[\,I_{n_1}\otimes H_2\mid H_1^{T}\otimes I_{m_2}\,],
\]
each row receives weight \(2\) from \(H_2\) and weight \((n_1-1)/2\) from the
column weight of \(H_1\), so
\(\max\operatorname{rowwt}(H_Z)=(n_1+3)/2\). Columns in the first block have weight at most \(2\), while columns in the
second block inherit the weight-three simplex checks,
which gives \(\max\operatorname{colwt}(H_Z)=3\).
\end{proof}

Theorem~\ref{thm:row_col_weight_bounds_polished} separates the roles of the
simplex dimension and the repetition length. Once \(n_2\ge3\), the maximum
\(X\)-check weight remains \(5\), and the maximum \(Z\)-column weight remains
\(3\), independently of \(n_2\). Increasing the repetition length therefore
raises the block length and the code distance, up to
\(n_2=2^{k-1}\), without increasing these local weights.

The dependence on \(k\) is different. Since \(n_1=2^k-1\), for \(k\ge3\) the
maximum \(X\)-column weight and \(Z\)-check weight are
\begin{equation}
\max\operatorname{colwt}(H_X)=2^{k-1}-1,
\qquad
\max\operatorname{rowwt}(H_Z)=2^{k-1}+1.
\end{equation}
Increasing \(k\) provides more logical qubits in each HGP block, while also
increasing the number of checks incident on some qubits and the size of the
largest \(Z\)-checks. This tradeoff directly affects syndrome-extraction
scheduling and decoding complexity.

For \(k=4\), the maximum row weights of \(H_X\) and \(H_Z\) are \(5\) and
\(9\), while the corresponding maximum column weights are \(7\) and \(3\).
These values remain unchanged as the repetition length increases from
\(n_2=3\) to \(n_2=8\). The \(k=4\) family therefore provides four logical
qubits per block while retaining moderate check and column weights. Its
distance-eight member also reaches the maximum distance \(2^{k-1}=8\)
supported by the simplex code, which motivates the \([[365,4,8]]\) instance
used as the primary circuit-level benchmark.

Table~\ref{tab:code_params_polished} lists representative members across
several values of \(k\) and \(d\). At fixed \(k\), increasing \(d\) changes
the block length and distance while leaving the check and column weights
unchanged for \(d\ge3\). The \(k=5,6\) examples illustrate the stronger growth
with the number of logical qubits per block.

\begin{table}[ht]
\begin{ruledtabular}
\begin{tabular}{ccccccc}
\(n\) & \(k\) & \(d\)
& \(\max\operatorname{rowwt}(H_X)\)
& \(\max\operatorname{colwt}(H_X)\)
& \(\max\operatorname{rowwt}(H_Z)\)
& \(\max\operatorname{colwt}(H_Z)\)
\\
\colrule
\(7\)   & \(2\) & \(2\) & \(4\) & \(2\)  & \(3\)  & \(3\) \\
\(21\)  & \(3\) & \(2\) & \(4\) & \(3\)  & \(5\)  & \(3\) \\
\(35\)  & \(3\) & \(3\) & \(5\) & \(3\)  & \(5\)  & \(3\) \\
\(49\)  & \(3\) & \(4\) & \(5\) & \(3\)  & \(5\)  & \(3\) \\
\colrule
\(65\)  & \(4\) & \(2\) & \(4\) & \(7\)  & \(9\)  & \(3\) \\
\(115\) & \(4\) & \(3\) & \(5\) & \(7\)  & \(9\)  & \(3\) \\
\(165\) & \(4\) & \(4\) & \(5\) & \(7\)  & \(9\)  & \(3\) \\
\(215\) & \(4\) & \(5\) & \(5\) & \(7\)  & \(9\)  & \(3\) \\
\(265\) & \(4\) & \(6\) & \(5\) & \(7\)  & \(9\)  & \(3\) \\
\(315\) & \(4\) & \(7\) & \(5\) & \(7\)  & \(9\)  & \(3\) \\
\(365\) & \(4\) & \(8\) & \(5\) & \(7\)  & \(9\)  & \(3\) \\
\colrule
\(1333\)  & \(5\) & \(8\)  & \(5\) & \(15\) & \(17\) & \(3\) \\
\(2821\)  & \(5\) & \(16\) & \(5\) & \(15\) & \(17\) & \(3\) \\
\(5061\)  & \(6\) & \(8\)  & \(5\) & \(31\) & \(33\) & \(3\) \\
\(10773\) & \(6\) & \(16\) & \(5\) & \(31\) & \(33\) & \(3\) \\
\(22197\) & \(6\) & \(32\) & \(5\) & \(31\) & \(33\) & \(3\) \\
\end{tabular}
\end{ruledtabular}
\caption{
Representative members of the simplex--repetition HGP family.
The \(k=4\), \(d=5,6,7,8\) rows are the instances used in the circuit-level
benchmark, with the \([[365,4,8]]\) code serving as the primary
representative. The \(k=5,6\) rows illustrate the growth of the block length
and of the dominant check and column weights as the number of logical qubits
per block is increased.
}
\label{tab:code_params_polished}
\end{table}


\appsection{app:numerical_details}{Numerical benchmark details}

We use the circuit-level noise model of
Ref.~\cite{Koh2026Phantom}, reproduced in Table~\ref{tab:noise}. The model
includes idle, one-qubit gate, two-qubit gate, measurement, and reset errors
with the relative rates listed below.

\begin{table}[t]
\centering
\begin{tabular}{lll}
\toprule
\textbf{Error mechanism} & \textbf{Model} & \textbf{Relative rate} \\
\midrule
Idle error & 1-qubit depolarizing & \(p/300\) \\
1-qubit gate error & 1-qubit depolarizing & \(p/15\) \\
2-qubit gate error & 2-qubit depolarizing & \(p\) \\
Measurement error & Classical bit-flip & \(5p/3\) \\
Reset error & Qubit bit-/phase-flip & \(p\) \\
\bottomrule
\end{tabular}
\caption{
Circuit-level noise model used in Ref.~\cite{Koh2026Phantom} and adopted here.
}
\label{tab:noise}
\end{table}

Both benchmarks compare the \(k=4\) simplex--repetition HGP family with
rotated surface-code baselines under the noise model in
Table~\ref{tab:noise}. Each HGP block encodes four logical qubits. In-block
logical CNOTs are implemented by physical qubit permutations and Pauli-frame
updates, while logical CNOTs between different HGP blocks are implemented
transversally. One QEC cycle is performed on the participating HGP blocks
after each inter-block transversal CNOT layer. The HGP syndrome-extraction
circuits are generated using the sign-balanced cardinal construction, and the
resulting detector error models are decoded using BP+LSD with LSD order zero.

For the surface-code baseline, each logical qubit is encoded in one rotated
surface-code block, and every logical CNOT is implemented by transversal
physical CNOTs between the corresponding blocks. The surface-code circuits
follow the same logical circuits and QEC cadence as their HGP counterparts.
Standard rotated surface-code syndrome-extraction circuits are generated using
\texttt{Stim}, and the resulting detector data are decoded with the
logical-observable MWPM decoder implemented in the open-source
\texttt{lomatching} package~\cite{Serra_Peralta_2026}. Both implementations
use ideal logical initialization and noiseless final logical readout.

For both benchmarks, let \(\mathcal S\) denote the task-specific set of
decoded final logical stabilizer observables with outcomes fixed by the ideal
noiseless circuit. We report the logical infidelity
\begin{equation}
    1-F
    =
    \frac{m}{N|\mathcal S|},
\label{eq:appendix_logical_infidelity}
\end{equation}
where \(N\) is the number of Monte Carlo shots and \(m\) is the total number
of decoded logical-observable outcomes that differ from their ideal values.
Thus the logical infidelity is the average mismatch probability per target
logical observable. Binomial standard errors are computed directly from the
same mismatch counts.

We consider the two tasks as stated in the main text.

\paragraph{Logical GHZ preparation.}
The target state is
\begin{equation}
    |\mathrm{GHZ}_K\rangle_L
    =
    \frac{|0\rangle_L^{\otimes K}+|1\rangle_L^{\otimes K}}{\sqrt{2}},
    \qquad
    K=8,16,32,64.
\end{equation}
The initial logical state is
\begin{equation}
    |+000\rangle_L
    |0000\rangle_L\cdots|0000\rangle_L.
\end{equation}
Within the first HGP block, a four-qubit logical GHZ seed is prepared using the
in-block CNOTs \(1\to2\), \(1\to3\), and \(2\to4\). The state is then
expanded across the remaining HGP blocks using inter-block transversal CNOT
layers arranged in a logarithmic-depth binary-tree schedule.

For each Monte Carlo shot, we decode the \(K\) logical GHZ stabilizers
\begin{equation}
\mathcal S_{\mathrm{GHZ}}
=
\{
\bar Z_1\bar Z_2,\,
\bar Z_2\bar Z_3,\,
\ldots,\,
\bar Z_{K-1}\bar Z_K,\,
\bar X_1\bar X_2\cdots\bar X_K
\}.
\end{equation}
All ideal outcomes are \(+1\). For this benchmark,
\(\mathcal S=\mathcal S_{\mathrm{GHZ}}\) and
\(|\mathcal S_{\mathrm{GHZ}}|=K\).

The plotted values are computed directly from the observable-mismatch counts.
For example, at \(p=10^{-3}\) and \(K=64\), the \([[365,4,8]]\) HGP code has
\(160\) mismatches out of \(1.28\times10^6\) decoded observables, giving a
logical infidelity of \(1.25\times10^{-4}\). The distance-eight surface-code
baseline has \(146143\) mismatches out of \(2.048\times10^8\) decoded
observables, giving a logical infidelity of \(7.14\times10^{-4}\). At
\(p=5\times10^{-4}\) and \(K=64\), the corresponding logical infidelities are
\(1.64\times10^{-5}\) for the HGP code and \(4.65\times10^{-5}\) for the
surface-code baseline.

\paragraph{Trotterized many-body simulation.}
We follow the circuit family of Ref.~\cite{Koh2026Phantom}. The logical circuit
simulates Trotterized time evolution generated by eight-body Ising interaction
terms interleaved with single-qubit transverse-field terms. Each interaction
layer contains terms of the form
\begin{equation}
    U_Z^{(j)}(\theta)
    =
    \exp\!\left(
    -i\theta\,
    \bar Z_j\bar Z_{j+1}\cdots\bar Z_{j+7}
    \right),
\end{equation}
together with the transverse-field layer
\begin{equation}
    U_X(\phi)
    =
    \prod_j \exp(-i\phi\bar X_j).
\end{equation}
We retain the compiled
forward-ladder--rotation--inverse-ladder circuit at the Clifford-compatible
angles \(\theta=\phi=\pi/2\), rather than algebraically simplifying the
resulting ideal unitary.  This preserves the entangling-gate structure of the
Trotter workload while enabling stabilizer-level noisy-circuit simulation.

The initial logical state is
\begin{equation}
    |++++0000\cdots ++++0000\rangle_L.
\end{equation}
The eight-body interactions are arranged in an open-boundary brickwork
pattern. The first interaction layer acts on the disjoint groups $\{1,\ldots,8\},\,
    \{9,\ldots,16\},\,
    \ldots,$
and the shifted layer acts on $ \{5,\ldots,12\},\,
    \{13,\ldots,20\},\,
    \ldots,$
with groups extending beyond the logical chain omitted. One Trotter cycle
consists of the unshifted interaction layer, the shifted interaction layer,
and the transverse-field layer. We simulate 2 Trotter cycles.

Each eight-body interaction is decomposed into a forward logical CNOT ladder,
a logical \(Z\) rotation on the parity qubit, and the inverse CNOT ladder. The
single-logical-qubit \(X\) and \(Z\) operations arising from the
Clifford-angle circuit are tracked in the logical Pauli frame and incorporated
into the expected outcomes of the final logical observables.

At the end of the circuit, we decode the complete set of logical stabilizer
observables
\begin{equation}
\mathcal S_{\mathrm{Trotter}}
=
\{
\bar X_1,\bar X_2,\bar X_3,\bar X_4,\,
\bar Z_5,\bar Z_6,\bar Z_7,\bar Z_8,\,
\ldots,\,
\bar Z_{K-3},\bar Z_{K-2},\bar Z_{K-1},\bar Z_K
\}.
\end{equation}
For each Monte Carlo shot, we compare their decoded outcomes with the noiseless
expected values, including the accumulated Pauli frame. For this benchmark,
\(\mathcal S=\mathcal S_{\mathrm{Trotter}}\) in
Eq.~\eqref{eq:appendix_logical_infidelity}.

We simulate two Trotter cycles for \(K=8,16,32\) logical qubits at \(p=10^{-3}\) and \(p=5\times10^{-4}\).  The
distance-eight HGP code improves over the distance-eight surface-code baseline
for all simulated instances.  At \(p=10^{-3}\) and \(K=32\), the logical
infidelities are \(1.39\times10^{-2}\) for the HGP code and
\(3.17\times10^{-2}\) for the surface-code baseline.  At
\(p=5\times10^{-4}\) and \(K=32\), the corresponding logical infidelities are
\(2.93\times10^{-3}\) and \(4.40\times10^{-3}\).

The separation is more modest than in the GHZ benchmark, as expected from the
structure of the workload.  In GHZ preparation, the in-block fanout seed is a
large part of the entangling structure and is directly absorbed into phantom
relabellings and Pauli-frame updates.  In the Trotterized circuit, by contrast,
many Trotter entangling layers and intervening QEC cycles remain active after
the in-block phantom relabellings are applied.  The benchmark therefore probes a
less favorable regime for phantomness than GHZ fanout, but the HGP data still
show a consistent circuit-level reduction in logical infidelity.  This indicates
that the phantom advantage is not limited to GHZ-style state preparation, while
also showing that its size is controlled by the fraction of the workload that
can be converted into relabellings.

\end{document}